\documentclass[letterpaper, 10 pt, conference]{ieeeconf}  

\IEEEoverridecommandlockouts                              
\newcommand{\x}{\bm x}
\newcommand{\barx}{\bar{\bm x}}
\newcommand{\bary}{\bar{\bm y}}
\newcommand{\barg}{\bar{\bm g}}
\newcommand{\Lambdab}{\bm \Lambda}
\newcommand{\xstar}{\bm{x}^*}
\newcommand{\G}{\bm G}
\newcommand{\A}{\bm A}

\newcommand{\M}{\bm M}
\newcommand{\y}{\bm y}
\newcommand{\Y}{\bm Y}
\newcommand{\C}{\bm C}
\newcommand{\R}{\bm R}
\newcommand{\vb}{\bm v}
\newcommand{\ub}{\bm u}

\newcommand{\1}{\bm 1}
\newcommand{\nl}{\left| \left| \left|}
\newcommand{\nr}{\right| \right|\right|}
\newcommand{\nru}{\right| \right|\right|_F^{\ub}}
\newcommand{\nrv}{\right| \right|\right|_F^{\vb^h}}
\newcommand{\diag}{\text{diag}}
\newcommand{\tr}{\text{tr}}

\usepackage{amsmath} 
\usepackage{amssymb}  
\usepackage[fleqn]{mathtools}
\usepackage{bm}
\usepackage{xcolor}
\usepackage{amsthm}
\usepackage{pgfplots}
\usepgfplotslibrary[groupplots]
\theoremstyle{plain}
\newtheorem{prop}{Proposition}  
\newtheorem{lemma}{Lemma}    
\newtheorem{theorem}{Theorem}  
\newtheorem{assum}{Assumption}   
    
\newtheorem{remark}{Remark} 
\usetikzlibrary{spy}
\allowdisplaybreaks

\title{\LARGE \bf
Gradient-Tracking over Directed Graphs for solving Leaderless Multi-Cluster Games
}

\author{Jan Zimmermann$^{1}$, Tatiana Tatarenko$^{1}$, Volker Willert$^{2}$ and J\"urgen Adamy$^{1}$
\thanks{This work was funded by the Deutsche Forschungsgemeinschaft (DFG, German Research Foundation) – SPP 1984, Project 16.}
\thanks{${1}$ Control Methods and Robotics Lab, Department of Electrical Engineering and Information Science,
        Technical University of Darmstadt, 64283 Darmstadt, Germany.
        {Email:\tt\small jan.zimmermann@rmr.tu-darmstadt.de}}%
\thanks{        ${2}$ Faculty of Electrical Engineering, University of Applied Sciences Würzburg-Schweinfurt,  97421 Schweinfurt, Germany.}
}

\begin{document}

\maketitle
\thispagestyle{empty}
\pagestyle{empty}

\begin{abstract}

We are concerned with finding Nash Equilibria in  agent-based multi-cluster games, where agents are separated into distinct clusters. While the agents inside each cluster collaborate to achieve a common goal, the clusters are considered to be virtual players that compete against each other in a non-cooperative game with respect to a coupled cost function. In such scenarios, the inner-cluster problem and the game between the clusters need to be solved simultaneously. Therefore, the resulting inter-cluster Nash Equilibrium should also be a minimizer of the social welfare problem inside the clusters. In this work, this setup is cast as a distributed optimization problem with sparse state information. Hence, critical information, such as the agent's cost functions, remain private. We present  a distributed algorithm that converges with a linear rate to the optimal solution. Furthermore, we apply our algorithm to an extended cournot game to verify our theoretical results.

\end{abstract}

\section{Introduction}
Nash Equilibrium problems arise in a broad variety of applications, such as wireless networks \cite{Charilas2010}, construction engineering and management \cite{Kapliski2010} and Smart Grids \cite{Saad2012}. In the latter, for optimal energy provision, Smart Grids are often represented as hierarchical models, where on a higher level an energy management problem needs to be solved between microgrids, while on a lower level, i.e. inside the microgrids, the distribution of energy or power needs to be optimally planned. The energy management problem is often cast as a non-cooperative game,  where the microgrids compete against each other regarding the energy price of a central power plant \cite{Atzeni13}, \cite{Kasbekar2012}. In contrast to that, lower level problems,  such as the economic dispatch problem \cite{Tatarenko2019distr}, \cite{Zimmermann2020}, are usually formulated as social welfare optimization problems, which depend on the result of the game on the higher level. 
Casting this problem as a multi-cluster game enables simultaneous solution of the non-cooperative game between microgrids and the cooperative distributed optimization inside the microgrids. Such an approach is likely to be more efficient than a separate solution on each level, as the result is optimal regarding both problems.\\
For the distributed  solution of multi-cluster games, most algorithms have been designed for continues time. The works of \cite{Ye2018} and \cite{Ye2017} aim to solve an unconstrained multi-cluster game by using gradient-based algorithms. Inspired by these results, authors of \cite{Ye2020} propose a gradient-free algorithm for a similar setup, where cost functions are unknown to agents and therefore, a payoff-based approach is used. All three publications do not define the inter-cluster communication and no explicit hierarchy between the agents inside a cluster is mentioned. In contrast to that, the following publications define the inter-cluster communication by undirected graphs and introduce a leader-follower hierarchy, in which only the cluster leader communicates with leaders from other clusters. 
In \cite{Yue2018} such leaders and followers exchange pseudo gradients to achieve the Nash Equilibrium of the considered constrained problem. The authors of  \cite{Zeng2019} and \cite{Zou2019} both employ gradient-free algorithms to a achieve a generalized Nash Equilibrium. \\
Less research has been dedicated to discrete-time setups. 
In the work of \cite{Meng2020} a leader-follower based algorithm for discrete-time settings is proposed, which can solve unconstrained multi-cluster games. 
To minimize the cost functions, a gradient-tracking approach is chosen.\\
All of the previous mentioned work deals with undirected communication architectures. \\
In this paper, we provide an algorithm that is based on the gradient-tracking algorithm of \cite{Pu2020} and solves the multi-cluster game. Each agent maintains two variables, one for the decision estimation of all agents and one for the gradient-tracking inside the cluster to which the agent belongs. The step-size of each gradient-step is considered to be constant, which is an advantage of gradient-tracking updates against other methods such as the one in \cite{Nedic2015}. In contrast to most of the mentioned publications, we consider a discrete time setting for our problem. Moreover, we define inter cluster communication, as it is done for example in \cite{Meng2020}.  Furthermore, we go beyond the leader-follower architecture of \cite{Meng2020}: In our approach the inter-cluster communication graph can be defined more generically as it is allowed that more than one agent can communicate with agents outside its own cluster. However, if the graph is defined such that only one agent sends and receives information from other clusters, we arrive at the leader-follower architecture of \cite{Meng2020}. Therefore, the hierarchy setup among agents of \cite{Meng2020} can be regarded as a special case of our approach. Furthermore, opposed to all mentioned literature concerning multi-cluster games, we consider directed communication, which generalize undirected architectures. Our contribution can therefore be summarized as follows: 1) We provide a discrete-time algorithm  that runs on directed, leader-free communication graphs and solves the distributed multi-cluster game. 2) We show convergence by approximation of the update equations of the algorithm with a linear, time-invariant state-space system as it is done in \cite{Meng2020} and \cite{Pu2020}. 3) At last, we verify our theoretical results with a simulation of an extended cournot game.

\section{Notation and Graphs}
All time indices $k$ belong to the set of non-negative integers $\mathbb{Z}^+$. Scalars are denoted by $x$, while we use boldface for vectors $\x$ and matrices $\A$. The expression $(x_i)_{i=1}^{n}$ vectorizes all $x_i$, i.e. $[x_1, ..., x_n]^T$. The same way, $(\x_i)_{i=1}^{n}$ stacks vectors to matrices.  We denote vector norms by $|| \cdot ||$ and matrix norms by $||| \cdot |||$. Agent $i$ is part of the considered agent system, consisting of $n$ agents. All agents are grouped  into $h = 1, ..., H$ clusters, where a cluster $h$ encompasses $n_h$ agents. The operator $\diag(\cdot)$ expands the vector $\x$ into a diagonal matrix with entries of $\x$ on its trace and $\diag\lbrace \A_1, ..., \A_n \rbrace$  expands a series of matrices $\A_i \in \mathbb{R}^{n_i \times q_i}$ into a block matrix with matrices $\A_1, ... ,\A_n$ on its diagonal as blocks and zero entries otherwise, such  that the resulting matrix is of dimension $\sum_{i} n_i\times \sum_i q_i$. For brevity of notation we use the following notation for gradients: $\nabla_{\x} f(\x,\y) \big|_{\x = \x_r} \triangleq \nabla_{\x_r} f(\x_r,\y)$.\\
A directed graph $\mathcal{G} = \lbrace \mathcal{V}, \mathcal{E} \rbrace$ contains a set of vertices $ \mathcal{V}$ and a set of edges $\mathcal{E} = \mathcal{V} \times \mathcal{V}$. Each vertex $v_i \in \mathcal{V}$ of a graph $\mathcal{G}$ is represented by an agent $i$ of the agent system and each edge $(j,i) \in \mathcal{E}$ is a directed communication channel from agent $j$ to agent $i$.
\section{Preliminaries}
\subsection{Communication}
In this section, we present the assumptions that we make towards the communication architecture of the agent system, which is described by graph theory. There are two separate communication layers: Layer one, which connects the agents inside the respective cluster but allows no communication to other clusters and layer two, which connects agents regardless of their cluster membership. For these, we make the assumptions that
\begin{assum} \label{as:graphC}
	The directed graph $\mathcal{G}^h$, which connects the agents inside cluster $h$, can be described by the weighted adjacency matrices $\C^h \in \mathbb{R}^{n_h \times n_h}$. It holds that
	\begin{itemize}
		\item all graphs $\mathcal{G}^h$, $h = 1, ..., H,$ are strongly connected, respectively, and
		\item the weights of each $\C^h$ are chosen such that the matrix is non-negative, column-stochastic, i.e. $\1^T \C^h = \1^T$, and has positive diagonal entries $\C^h_{ii}$. 
	\end{itemize}
\end{assum}
\begin{assum} \label{as:graphR}
	The directed graph $\mathcal{G}$, which models connections both inside the clusters as in-between, can be described by the adjacency matrix $\R \in \mathbb{R}^{n \times n}$. It holds that
	\begin{itemize}
		\item the graph is strongly connected and
		\item the weights of $\R$ are chosen such that the matrix is non-negative, row-stochastic, i.e. $\R \1  = \1$, and has positive diagonal entries $\R_{ii}$. 
	\end{itemize}
\end{assum}
\begin{remark}
	For column-stochastic weighting see for example Remark 2 of \cite{Zimmermann2020}. For row-stochastic weighting use $\R_{ij} = 1/ \delta_j^+$, $i = 1, ..., n$, where  $\delta_j^+$ is the in-degree of node $j$. The in-degree can be determined by simple message forwarding.
	
\end{remark}
Some properties of the eigenvectors of the communication matrices can be summarized  in the following Lemma from \cite{Pu2020}:
\begin{lemma} \label{lemma:RCeigen}
	Let matrices $\C^h$ and $\R$ be defined as in Assumptions \ref{as:graphC} and \ref{as:graphR}, respectively. Let all statements of these assumptions hold. Then,
	\begin{itemize} 
		\item the matrices $\C^h$, $h = 1,..., H$, each have a unique, positive right eigenvector $\vb^h$ with regard to eigenvalue $1$, normed such that  $\1^T \vb^h = 1$, i.e., it holds that $ \C^h \vb^h = \vb^h$,
		\item  the matrix $\R$ has a unique, positive left eigenvector $\ub$, normed such that $\ub^T \1 = 1$, i.e., it holds that $\ub^T \R = \ub^T$. 
	\end{itemize}
\end{lemma}
\subsection{Matrix norms} \label{subsec:matrixnorms}
In the theoretical part of this work, we will need definitions for weighted spectral matrix norms, for which we take orientation from \cite{Xin2019} and weighted Frobenius matrix norms, for which our results are loosely based on \cite{Meng2020}. \\
We define the weighted spectral matrix norms for arbitrary quadratic matrices $\bm X \in \mathbb{R}^{r \times r}$ as follows:
\begin{align*}
||| \bm X |||_2^{\ub}  &\triangleq ||| \diag(\sqrt{\ub}) \bm X \diag(\sqrt{\ub})^{-1}|||_2, \\
||| \bm X |||_2^{\vb^h}  &\triangleq ||| \diag(\sqrt{\vb^h})^{-1} \bm X \diag(\sqrt{\vb^h})|||_2, 
\end{align*}
 using the left eigenvector $\ub$ from matrix $\R$, defined in Assumption \ref{as:graphR}, and right eigenvector $\vb^h$ from matrix $\C^h$, defined in Assumption \ref{as:graphC}. 
Note, that these definitions correspond to the norms in equation (4) and (5) of \cite{Xin2019}. \\
The Frobenius inner product for real, rectangular matrices $\bm A, \bm B \in \mathbb{R}^{r \times s}$ is defined as $\langle \bm \A, \bm B \rangle_{F} = \tr(\bm B^T \bm A) $, see \cite{Horn}. 
The weighted Frobenius inner products $\langle \bm \A, \bm B \rangle_{F}^{\ub} = \tr(\bm B^T \diag(\ub) \bm A), \langle \bm \A, \bm B \rangle_{F}^{\vb^h} = \tr(\bm B^T \diag(\vb^h)^{-1} \bm A)$
induce the weighted Frobenius matrix norms
\begin{align}
||| \A |||_F^{\ub} & \triangleq ||| \diag(\sqrt{\ub})\A|||_F, \\
||| \A |||_F^{\vb^h} & \triangleq  ||| \diag(\sqrt{\vb^h})^{-1}\A|||_F.
\end{align}
Based on the equivalence of norms, it can be established that there exist constants $\delta_{u,F}, \delta_{v, F}, \delta_{F, v}, \delta_{F, u} > 0$ such that
\begin{align*}
&||| \cdot |||^{\ub}_F \leq \delta_{u,F} ||| \cdot |||_F, &  ||| \cdot |||^{\vb^h}_F \leq \delta_{v,F} ||| \cdot |||_F, \\
&||| \cdot |||_{F} \leq \delta_{F,u} ||| \cdot |||_F^{\ub}, &  ||| \cdot |||_{F} \leq \delta_{F,v} ||| \cdot |||_F^{\vb^h}.
\end{align*}
Creating such upper bounds is standard in relevant literature \cite{Meng2020},  \cite{Pu2020},  \cite{Xin2019}.\\
Concerning the standard Frobenius norm of a matrix product, the following upper bound can be provided using the spectral matrix norm: 
\begin{lemma}\label{lemma:submult}
	For arbitrary matrices $\bm A \in \mathbb{R}^{n \times n}$, $\bm B \in \mathbb{R}^{n \times q}$, it holds that $\nl \bm A \bm B \nr_F \leq \nl \bm A \nr_2 \nl  \bm B \nr_F$.
\end{lemma}
This Lemma can be proved using the submultiplicative property of the Frobenius norm. We skip the mathematical details for brevity. \\
The matrix norms defined above can now be used in the following Lemma
\begin{lemma}\label{lemma:sigma}
	Let Assumptions \ref{as:graphC} and \ref{as:graphR}  hold and matrices $\R$ and $\C^h$ be defined as therein. The vectors $\ub$ and $\vb^h$ are the respective eigenvectors. Then, for arbitrary $\x \in \mathbb{R}^{n \times q}$ and $\y \in \mathbb{R}^{n_h \times q_h}$, there exist positive constants $\sigma_R, \sigma_C < 1$ such that
	\begin{align}
	\nl \R \x - \1 \ub^T \x\nr_F^{\ub} &\leq  \sigma_R \nl \x - \1 \ub^T \x\nr_F^{\ub}, \label{eq:lemmasigmaR} \\
	\nl \C^h \y -  \vb^h \1^T \y\nr_F^{\vb^h} &\leq  \sigma_C \nl \y - \vb^h \1^T  \y\nr_F^{\vb^h}. \label{eq:lemmasigmaC}
	\end{align}
\end{lemma}
This Lemma is an adjusted version of Lemma 4 of \cite{Pu2020}. Again, the details of the proof are skipped for brevity. 
%

\section{Multi-cluster games and gradient-tracking}
\subsection{Problem formulation}
Assume an agent system consisting of $n$ agents, in which each agent has communication and computation abilities. The storage capacity of these agents is limited. The agents are grouped into $H$ clusters and $n_h$ agents belong to cluster $h$. The set  $\mathcal{A}_h$ contains the agents of cluster $h$. All of these sets for clusters $h = 1, ..., H$ are disjoint, $\mathcal{A}_i  \cap \mathcal{A}_j = \emptyset$ for $i \neq j$.  The agents are connected by two different communication graphs. The graph $\mathcal{G}^h$ connects only the agents inside cluster $h$, while the graph $\mathcal{G}$ connects agents independently of their cluster membership. Therefore, $\mathcal{G}^h$ restricts the communication to intra-cluster exchange of information, while $\mathcal{G}$ enables global messaging.\\
It is assumed that each cluster forms a coalition, which means that all agents inside a cluster collaborate to achieve a common goal. In contrast to this, the clusters compete against each other regarding some coupled cost function. This inter-cluster competition can be modeled as a non-cooperative game, where the clusters are regarded as virtual players, while the actual decisions and actions are determined by the agents inside the clusters. Mathematically, we model this setup as follows.\\
Agent $i$ of cluster $h$ has exclusive access to its personal cost function $f_i^h(\x)$, which is assumed to be convex. This cost function is only known by agent $i$ and unknown to every other agent, independently of cluster membership. Vector $\x \in \mathbb{R}^q$ is the shared decision vector that can be separated into decisions of cluster $h$, i.e. $\x^h \in \mathbb{R}^{q_h}$, and the decisions of all other clusters, which are denoted by $\x^{-h} \in \mathbb{R}^{q - q_h}$. The cluster cost function $F^h:\mathbb{R}^q \rightarrow \mathbb{R}$ of cluster $h$ is declared as
\begin{equation}\label{eq:gameformulation}
	F^h(\x^h, \x^{-h}) = \sum_{i = 1}^{n_h} f_i^h(\x^h, \x^{-h}).
\end{equation}
Each cluster aims to minimize this function $F^h$, which depends not only on the actions $\x^h$ of cluster $h$ but also on the actions of all other clusters. However, in the optimization process, agents can only adjust the decisions of their own cluster while observing $\x^{-h}$.\\
With all considerations from above, we can express the multi-cluster game $\Gamma(H, \mathbb{R}^q, \lbrace F^h\rbrace)$, with the clusters as virtual players, as the following optimization problem \footnote{Note that if the number of clusters is reduced to $H = 1$, we arrive at the standard definition of an unconstrained, distributed optimization problem as described in \cite{Nedic2015} or \cite{Pu2020}.} :
\begin{equation}\label{eq:gameproblem}
	\min_{\x^h \in \mathbb{R}^{q_h}} F^h(\x^h, \x^{-h}) = \min_{\x^h \in \mathbb{R}^{q_h}} \sum_{i=1}^{n_h} f_i^h(\x^h, \x^{-h}),
\end{equation} 
$\forall h = 1, ..., H.$ In order to evaluate the gradient of the local cost function, an estimation of other cluster's decisions needs to be present. Therefore, every agent $i$ maintains a vector $\x_i$ that estimates the decisions of all clusters in the network.
In order for a solution $\xstar$ to be an optimum of the defined game, the following conditions need to be fulfilled:
\begin{itemize}
	\item Consensus among agents concerning the local state estimations: \vspace{-0.3cm} \begin{equation}
		\x_i = \x_j = \xstar, \ i,j = 1, ..., n,  \label{eq:consensuscondition}
	\end{equation}
	\item Social welfare minimum inside all clusters $h= 1, ..., H$ for sum of convex functions: \begin{equation}
	\sum_{i=1}^{n_h} \nabla_{\x^h} f_i^h((\x^h)^*, (\x^{-h})^* ) = 0 \label{eq:socialwelfarecondition}
	\end{equation}

	\item Nash Equilibrium for game $\Gamma(H, \mathbb{R}^q, \lbrace F^h\rbrace)$ between clusters:
	\begin{equation}
	F^h\left((\x^h)^*, (\x^{-h})^*\right) \leq F^h\left(\x^h, (\x^{-h})^*\right), \forall h \label{eq:nasheqcondition}.
	\end{equation}
\end{itemize}
The consensus condition is necessary, because the final decision vectors need to be the same at every agent. Every estimation should converge to the optimal decision $\xstar$ that satisfies the social welfare and Nash Equilibrium conditions.\\
Before we describe our algorithm, which solves the formulated problem, we first make assumptions regarding the local cost functions and their gradients to further specify the class of problems that we consider.
\begin{assum}\label{as:lipschitz}
	All local cost functions $f_i^h(\x^h, \x^{-h})$, for all $i = 1, ..., n_h$ and $h = 1, ..., H$, are convex,  continuously differentiable  and the gradient $\nabla_{\x^h} f_i^h(\x^h, \x^{-h})$ is Lipschitz continuous on $\mathbb{R}^{q_h}$, i.e. there exist constants $L_i^h > 0$ such that
	\begin{align*}
		||\nabla_{\x^h} f_i^h(\x^h, \x^{-h})  - \nabla_{\tilde{\x}^h} f_i^h(\tilde{\x}^h, \tilde{\x}^{-h}) ||_2& \\
		\leq L_i^h || \x^h - \tilde{\x}^h ||_2   \leq L_i^h || \x - \tilde{\x} ||_2.&
	\end{align*}
	Furthermore, it can be assumed that there exists a constant $L > 0$ such that
	$L \geq L_i^h, \ \forall i, h$.
\end{assum}
This assumption is standard in gradient-based distributed optimization \cite{Nedic2015}, \cite{Pu2020}. Next, we define	$\bm g^h(\x^h, \x^{-h}) \triangleq \sum_{i = 1}^{n_h} \nabla_{\x^h}f_i^h(\x^h, \x^{-h}) \in \mathbb{R}^{q_h}$
and the game mapping $ \bm M: \mathbb{R}^q \rightarrow \mathbb{R}^q$
\begin{equation} \label{eq:gamemapping}
	\M(\x) = [\bm g^1(\x)^T, ..., \bm g^H(\x)^T]^T.
\end{equation}
We make the assumption that
\begin{assum}\label{as:strongmonotonegame}
	The game mapping $M(\x)$ is strongly monotone on $\mathbb{R}^q$ with constant $\mu > 0$. 
\end{assum}
\begin{remark}
	This assumption is necessary for uniqueness of the Nash Equilibrium of game $\Gamma(H, \mathbb{R}^q, \lbrace F^h\rbrace)$.  Note that with this assumption it holds $\forall \x, \y \in \mathbb{R}^q$ that
	\begin{align*}
		\sum_{h = 1}^H  \left[\left( \sum_{i = 1}^{n_h} \left(\nabla_{\x^h} f_i^h(\x) - \nabla_{\y^h}f_i^h(\y)\right)\right)^T(\x^h - \y^h) \right]\\
		\geq \mu \sum_{h = 1}^H ||\x^h - \y^h||_2^2 = \mu ||\x - \y||_2^2.
	\end{align*}
\end{remark}

\subsection{Algorithm}
Let $k = 1, 2, ...$ be the time index. At each instance $k$ vector $\x_i(k)$ contains agent $i$'s estimations of the decisions of all clusters and therefore takes the form 
\begin{equation*}
	\x_i(k) = \left[ (\x_i^1(k))^T, ..., (\x_i^h(t))^T, ..., (\x_i^H(k))^T \right]^T \in \mathbb{R}^q,
\end{equation*}
where $\x_i^h(k) \in \mathbb{R}^{q_h}$ is the estimation made by agent $i$ of cluster $h$'s decisions at time $k$. The estimation vectors of all $n$ agents in the system can then be stacked to receive a matrix of all estimations at time $k$
\begin{equation*}
	\x(k) = [\x_1(k), ..., \x_N(k)]^T \in \mathbb{R}^{n \times q}. 
\end{equation*}
Without loss of generality, we assume that the row sequence of $\x(k)$ is ordered according to the numbering of cluster $h= 1, ..., H$. This means that the first $n_1$ rows of $\x(k)$ are estimations of agents that belong to cluster $1$, the next $n_2$ rows are estimations of agents belonging to cluster $2$ and so on.\\
Furthermore, we introduce the variable $\y_i^h(k) \in \mathbb{R}^{q_h}$, which is agent $i$'s local tracking variable of the gradient in cluster $h$ at time $k$. Assuming that $i \in \mathcal{A}_h$, we define the vector 
\begin{equation*}
	\hat{\y}_i^h(k) = [\bm 0_{1\times n_{<h}}, (\y_i^h(k))^T, \bm 0_{1\times n_{>h}}]^T \in \mathbb{R}^q,
\end{equation*}
with $n_{<h}\sum_{l = 1}^{h-1} n_l$ and $n_{>h}\sum_{l = h+1}^{H} n_l$. We  stack all local tracking variables of cluster $h$
\begin{equation*}
	\y^h(k) = [\y^h_1(k), ..., \y^h_{n_h}(k)]^T \in \mathbb{R}^{n_h \times q_h}
\end{equation*}
and then include all tracking variables of the separate clusters in the block matrix
\begin{equation*}
	\Y(k) = \diag\lbrace\y^1(k), ..., \y^H(k)\rbrace \in \mathbb{R}^{n \times q}.
\end{equation*}
It is important to initialize all local tracking variables $\y_i^h(0)$ with the local gradient at starting estimation $\x_i(0)$, i.e.
\begin{equation} \label{eq:yhinit}
	\y_i^h(0) = \nabla_{\x_i^h} f_i^h(\x_i(0)).
\end{equation}
At last, we define matrix $\G^h(k) \in \mathbb{R}^{n_h \times q_h}$, which contains the gradients of cluster $h$ at time $k$:
\begin{equation*}
	\G^h(k) = [\nabla_{\x_1^h} f_1^h(\x_1(k)), ..., \nabla_{\x_{n_h}^h} f_{n_h}^h(\x_{n_h}(k))]^T.
\end{equation*}
With all above definitions, we are able to formalize our algorithm with agent-wise update equations as follows:
\newpage
\begin{subequations}\label{eq:algorithmagent}
	\begin{align}
		\x_i(k+1) &= \sum_{j=1}^n \R_{ij}(\x_i(k) - \alpha \hat{\y}_i^h(k)) , \\
		\y_i^h(k+1) &= \sum_{j=1}^{n_h} \C^h_{ij} \y_j^h(k) \\
		&+ 	\nabla_{\x_i^h} f_i^h(\x_i^h(k+1), \x_i^{-h}(k+1) ) \nonumber \\
		&-\nabla_{\x_i^h} f_i^h(\x_i^h(k), \x_i^{-h}(k)) \nonumber,
	\end{align}
\end{subequations}
where $\alpha$ is a positive, constant step-size. \\
Using the stacked vectors and matrices from above, we can write a matrix-update representation of the algorithm:
\begin{subequations}\label{eq:algorithmvector}
	\begin{align}
		\x(k+1) &= \R \left(\x(k) - \alpha \Y(k) \right) \label{alg:vector_x}, \\
		\y^h(k+1) &= \C^h \y^h(k) +  \G^h(k+1)  - \G^h(k)  \label{alg:vector_y}.
			\end{align}
\end{subequations}
This algorithm is based on the push-pull algorithm, discussed for example in \cite{Pu2020} or \cite{Xin2019}, which was adapted to the cluster game scenario. 
However, in contrast to the algorithms in these publications, in the gradient tracking step of our algorithm, all agents solely exchange information with other agents of the same cluster $h$. This is ensured by setting the graphs $\mathcal{G}^h$ and therewith the matrix $\C^h$ appropriately. Furthermore, in the update step of the estimation, only those estimations of decisions are updated by the gradient information that belong to the respective cluster $h$ that $i$ is part of. Estimations $\x_i^{-h}$ of agent $i$ are updated without using gradient information. By this structure, we can assure that own decisions are pushed towards the local social welfare optimum in the respective clusters, while estimations of the decisions of other clusters are pushed towards a consensus.\\
In contrast to the algorithm presented in \cite{Meng2020}, our algorithm relies on directed communication, which extends the range of applications. Furthermore, in \cite{Meng2020} only leaders can exchange information between clusters. By allowing more inter-cluster communication channels in our approach, the amount of exchanged information  can be increased such that an inter-cluster consensus is likely to be achieved faster. For a preliminary convergence comparison, see the end of the Simulation section. 

\subsection{Convergence}
Let the vectors $\barx(k) \in \mathbb{R}^{1 \times q}$ and $\bary^h(k) \in \mathbb{R}^{1 \times q_h}$ be defined as
\begin{align*}
	\barx(k) =  \ub^T \x(k), \qquad \bary^h(k) = \frac{1}{n_h} \1^T \y^h(k),
\end{align*} 
with $\ub$ being the left eigenvector of matrix $\R$, see Lemma \ref{lemma:RCeigen}.
Due to initialization of $\y^h$ in Equation \eqref{eq:yhinit}, it can be shown by induction that 
\begin{equation}\label{eq:yinduction}
	\bary^h(k) = \frac{1}{n_h}  \sum_{i = 1}^{n_h} \nabla_{\x_i^h} f_i^h(\x_i^h(k), \x_i^{-h}(k)).
\end{equation}
By using the eigenvectors $\vb^h$ of matrices $\C^h$, we define the block matrices
\begin{align}
	\Lambdab_{\y} &= \diag\lbrace\vb^1\bary^1(k), ..., \vb^H\bary^H(k)\rbrace \in \mathbb{R}^{n \times q}, \\
	\Lambdab_{\barg} &=  \diag\lbrace\vb^1\barg^1(k), ..., \vb^H\barg^H(k)\rbrace \in \mathbb{R}^{n \times q},
\end{align}
using the vectors \begin{equation*}
		\barg^h (k) = \frac{1}{n_h}  \sum_{i = 1}^{n_h} (\nabla_{\barx^h}f_i^h(\barx^h(k), \barx^{-h}(k)))^T \in \mathbb{R}^{1 \times q^h}.
\end{equation*}

The general structure of our convergence proof is a known procedure (see \cite{Meng2020}, \cite{Pu2020}), which we extend to our problem formulation. The main idea is to show that the norms $\nl \1 \barx(k)- \1 \xstar  \nr_F, \nl \x(k) - \1 \barx(k) \nru$ and $\sum_{h=1}^H\nl \y^h(k) - \vb^h \bary^h(k) \nrv $converge to zero as time goes to infinity when using the update equation of $\eqref{eq:algorithmagent}$ or \eqref{eq:algorithmvector}. This means that all estimations $\x_i$ and all tracking variables $\y_i^h$ converge to a stable state. This can be achieved by upper bounding the update steps by a linear, time-invariant matrix $\A$ and showing that the spectral radius of this matrix is strictly smaller than 1. Therefore, we show that:
%
%
\begin{prop}\label{prop:matrixA}
	Let Assumptions \ref{as:graphC}, \ref{as:graphR} and \ref{as:lipschitz}  hold. The vector $\xstar$ fulfills the optimality condition \eqref{eq:socialwelfarecondition}. Using update equations of the Algorithm in \eqref{eq:algorithmagent} or \eqref{eq:algorithmvector}, the following linear inequality system can be established:
	\begin{align}
	\begin{bmatrix}
		\nl \1 \barx(k+1)- \1 \xstar  \nr_F \\
		\nl \x(k+1) - \1 \barx(k+1) \nru \\
		\sum_{h=1}^H\nl \y^h(k+1) - \vb^h \bary^h(k+1) \nrv
	\end{bmatrix} \\ \leq \A  \begin{bmatrix}
	\nl \1 \barx(k)- \1 \xstar  \nr_F \\
	\nl \x(k) - \1 \barx(k) \nru \\
	\sum_{h=1}^H\nl \y^h(k) - \vb^h \bary^h(k) \nrv
	\end{bmatrix},
	\end{align}
	with matrix
	\begin{align*}
		\bm A = \begin{bmatrix}
		\phi(\alpha) & \alpha a_{12} & \alpha a_{13} \\
		\alpha a_{21}&\sigma_R + \alpha a_{22} & \alpha a_{23} \\
		\alpha a_{31 }& a_{32} + \alpha a'_{32} & \sigma_c + \alpha a_{33}
	\end{bmatrix}
		\end{align*}

and scalar function
\begin{equation*}
	\phi(\alpha) = 	\sqrt{1  - 2\alpha \underline{\eta} \mu +  \alpha^2 L_v^2 \nl \1 \ub^T \nr_2^2}
\end{equation*}
where $L_v = L \max_{i,h} \lbrace v_i^h\rbrace$ and  $0 < \underline{\eta} \leq \min_h \lbrace \eta^h = \frac{(\vb^h)^T \ub^h}{n_h} \rbrace $.\\
All factors $a_{12}, a_{13}, a_{21}, a_{22}, a_{23}, a_{31}, a_{32},a'_{32}, a_{33}$ are positive. 
	
\end{prop}
The proof of this Proposition relies on the argumentation in \cite{Meng2020}, which we were able to adjust to our setting, i.e. incorporation of row- and column-stochastic matrices $\R$ and $\C^h$ that allow for more general communication than undirected leader-follower architectures. For the mathematical details and factors $a_{ij}$ see Appendix \ref{ap:proposition}. 
\newline
Now that we have the upper bound of every iteration, we need to show that system matrix $\A$ is stable. Therefore, we demonstrate that
\begin{prop}\label{prop:Astable}
	There exists a step-size $\alpha > 0$ such that the spectral radius of $\bm A(\alpha)$ is strictly smaller than $1$, i.e. 
	\begin{equation}
		\rho(\bm A) < 1.
	\end{equation}
\end{prop}
\begin{proof}
	Again, we take our orientation from \cite{Meng2020}. For $\alpha = 0$, matrix $\bm A$ contains the entries
	\begin{equation}
		\bm A(\alpha = 0) = \begin{bmatrix}
			1 & 0 & 0 \\
			0&\sigma_R  & 0\\
			0& a_{32} & \sigma_C
		\end{bmatrix}.
	\end{equation}
	Because of $0 < \sigma_R< 1$ and $ 0<  \sigma_C < 1$, see Lemma \ref{lemma:sigma},  it holds that  $\rho(\bm A(0))  = 1$ and $\bm A(0)$ has the  right eigenvector $\ub = [1 \  0 \  0]^T$ corresponding to eigenvalue $\lambda_1 = 1$. Now we need to investigate how this eigenvalue $\lambda_1$ changes, when the value of $\alpha$ increases from $0$. For this, the eigenvalue problems provides us with
	$	\frac{d \lambda_1(\alpha)}{d \alpha}\big|_{\alpha = 0} \ub =\frac{d  \bm A(\alpha)} {d \alpha}\big|_{\alpha = 0} \ub,$
	from which 
	\begin{align*}
		\frac{d \lambda_1(\alpha)}{d \alpha}&\Bigg|_{\alpha = 0} =  \frac{d \phi(\alpha)}{d \alpha } \Bigg|_{\alpha = 0} \\
		& = \frac{- 2 \underline{\eta} \mu +  2\alpha L_v^2 \nl \1 \ub^T \nr_2^2}{2 \sqrt{1  - 2\alpha \underline{\eta} \mu +  \alpha^2 L_v^2 \nl \1 \ub^T \nr_2^2}}\Bigg|_{\alpha = 0} 
		= -  \underline{\eta} \mu
	\end{align*}
	follows, where $0 < \underline{\eta} \leq \eta^h = \frac{(\vb^h)^T \ub^h}{n_h}   \ \forall h$, because all $\vb^h, \ h = 1, ..., H$, and $\ub$ are positive, see Lemma \ref{lemma:RCeigen}. Therefore, $-  \underline{\eta} \mu < 0$. Because of this, the value of the spectral radius $\rho(\bm A(\alpha))$ decreases as $\alpha$ increases from zero. Following from this fact, together with the continuity of the evolution of spectral radii, there must exist an $\alpha > 0$ for which $\rho(\bm A(\alpha))< 1$. 
\end{proof}
With this result, we now know that the variables of the algorithm \eqref{eq:algorithmagent} or \eqref{eq:algorithmvector} converge to stable states. In the following theorem we combine all of the results above and show the problem of the multi-cluster game can be solved with our algorithm:
\begin{theorem}
	Let Assumption \ref{as:graphC}, \ref{as:graphR}, \ref{as:lipschitz} and \ref{as:strongmonotonegame}  be given. Then, there exists a unique Nash Equilibrium for the game $\Gamma(H, \mathbb{R}^q, \lbrace F^h\rbrace)$ defined in \eqref{eq:gameproblem}. Furthermore, using the update equations in \eqref{eq:algorithmagent} or \eqref{eq:algorithmvector}, it holds that the estimations of all agents reach a consensus
	\begin{equation}\label{eq:toinfinityandbeyond}
		\lim_{k \rightarrow \infty}\x_j(k) =  \lim_{k \rightarrow \infty}\x_i(k) = \xstar , \forall i,j = 1, ..., n,
	\end{equation}
	where the stable state $\xstar$ fulfills optimality conditions \eqref{eq:socialwelfarecondition} and \eqref{eq:nasheqcondition}. With this, optimality condition \eqref{eq:consensuscondition} is fulfilled as well. Therefore, the multi-cluster game is solved and the convergence rate to the optimum is linear. 
\end{theorem}

\begin{proof}
	Given strong monotony of the mapping $\M(\x)$ as claimed in Assumption \ref{as:strongmonotonegame}, there exists a unique Nash Equilibrium and the vector $\xstar$ is this unique stable point, if it satisfies $\M(\xstar) = \bm 0$ \cite{Tatarenko2019}.
	From Propositions \ref{prop:matrixA} and \ref{prop:Astable}, we know that all estimations converge linearly to a consensus and this consensus is $\xstar$, which satisfies optimality condition \eqref{eq:socialwelfarecondition}. Therefore, the expression in Equation \eqref{eq:toinfinityandbeyond} is true and the convergence rate is linear.\\ 
	 According to the definition of the game mapping in Equation \eqref{eq:gamemapping}, $\M(\xstar) = \bm 0$ holds if $\bm g^h(\xstar) \triangleq \sum_{i = 1}^{n_h} \nabla_{\x^h}f_i^h(\xstar) = \bm 0$ for all $h = 1, ..., H$. This in turn is the condition for the local social welfare optimum in Equation \eqref{eq:socialwelfarecondition}. This means that if the estimations of all agents are in consensus with each other, i.e. condition \eqref{eq:consensuscondition} holds, and this consensus vector $\xstar$ fulfills condition \eqref{eq:socialwelfarecondition} for all $h = 1, ..., H$, then the Nash Equilibrium condition \eqref{eq:nasheqcondition} is always fulfilled as well. We showed convergence towards $\x^*$, which fulfills condition \eqref{eq:socialwelfarecondition}, in Propositions \ref{prop:matrixA}  and \ref{prop:Astable}. Therefore, the proof is concluded.
\end{proof}
\section{Simulation}
For the verification of our algorithm by simulation, we chose a variant of the well-known cournot game, like it is done in \cite{Meng2020}. In the scenario there are $n$ factories that produce the same product. We describe the amount of product units produced by factory $i$ with the scalar $x_i$. Each factory has an individual cost $C_i(x_i)$ for producing $x_i$ units of the product. For our simulation, we chose $	C_i(x_i) = a_i x_i^2 + b_i x_i + c_i$. By selling $x_i$ units for a price $P(\x)$, which depends on the vectorized output $\x = [x_1, ..., x_n]^T$ of all factories, each factory generates revenue. We chose the price function $P(\x) = P_c - \sum_{j=1}^n x_j$, where $P_c$ is some positive constant. With this, factory $i$ has the objective function $f_i(\x) = C_i(x_i) - x_iP(\x).$
Each of the factories belongs to a company $h$ and this company aims to minimize the objective function for all of the factories that belongs to it
\begin{equation}\label{eq:costcournot}
	\min_{\x^h} F^h = \min_{\x^h} \sum_{i = 1}^{n_h} C_i(x_i) + x_iP(\x)
\end{equation} 
by adjusting the output $\x^h = [(x_i)_{i=1}^{n_h}]$, while competing against other companies. Thereby, we have arrived at the formulation of a non-cooperative, multi-cluster game, where the factories correspond to agents and the companies are the agent-containing clusters.\\
It can readily be confirmed that cost functions in Equation \eqref{eq:costcournot} fulfill Assumptions \ref{as:lipschitz} and \ref{as:strongmonotonegame}. 
We simulated our algorithm with $H = 3 $ companies. Company $h=1$ owns four factories, while companies $h=2,3$ each own three. The factories are connected by inner-cluster graphs $\mathcal{G}^h$ and a global graph $\mathcal{G}$, which were chosen such that they satisfy Assumptions \ref{as:graphC} and \ref{as:graphR}, respectively.  The starting estimations of the decisions of each agent were chosen randomly and the tracking variables were set according to Equation \eqref{eq:yhinit}. The step-size was set to $\alpha = 0.1$ . The results of the simulation are shown in Figure \ref{fig:cournot_states}. 
\begin{figure}[h]
	\centering
	\scalebox{0.9}{
		\input{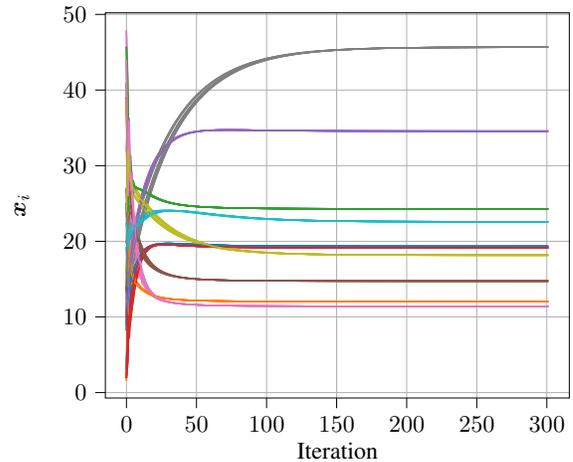} }
	\caption{Convergence of estimations $\x_i$ to the optimum $\xstar$ for the multi-cluster cournot game with 10 factories. Each color represents a different dimension of the decision vector $\x$ and each line is the estimation of one dimension made by one agent.}
	\label{fig:cournot_states}
\end{figure}
It can be seen that a consensus is reached in every dimension of the final decision vector. In fact, it can be confirmed that after 300 iterations the sum of absolute differences between the estimation of agent $1$, chosen as a representative, and all other agents is less than $1.90\times 10^{-3}$ for all dimensions of $\x$. Furthermore, the normed difference of the estimation $\x_1$ of agent $1$ to the Nash Equilibrium state $\xstar$ can be calculated with $ \epsilon = ||\x_1 - \xstar||_2 =  0.005$. When restricting the setup to a leader-follower hierarchy, such as in \cite{Meng2020}, where only the cluster leader communicates with other clusters, we need about 500 iterations to receive a comparable accuracy with the same parameterization.

\section{Conclusion}
Within this work, we presented a distributed algorithm for the solution of a class of multi-cluster games and we proved linear convergence to an optimal decision vector that fulfills the optimality conditions of the Nash Equilibrium between the clusters and  social welfare optimum inside each cluster. As less restrictions are imposed on the communication architecture, the algorithm is applicable to a wider range of problems than the one in \cite{Meng2020}. However, in order to handle more complicated scenarios such as the multi-level energy provision problem in Smart Grids, mentioned in the introduction, the optimization procedure needs to be able to respect constraints of agents or clusters. How to include such constraints into our algorithm seems to be a promising field of future research. 

\bibliography{clusterbib}
\bibliographystyle{plain}

\appendix
\subsection{Supporting Lemma for proof of Proposition \ref{prop:matrixA}}
\begin{lemma}\label{lemma:inequalities}
	Let $L_v = L \max_{i,h} \lbrace v_i^h\rbrace$. $\xstar \in \mathbb{R}^{1 \times q}$ is a solution of the problem that fulfills optimality condition \eqref{eq:socialwelfarecondition}. Then, the following inequalities hold:
	\begin{align*}
	1) \quad&\nl \1 \barx(k) - \1 \xstar -  \1 \ub^T \alpha \Lambdab_{\barg} \nr_F  \\ 
	&\leq  \sqrt{1  - 2\alpha \underline{\eta} \mu +  \alpha^2 L_v^2 \nl \1 \ub^T \nr_2^2}\nl \1 \barx(k) - \1 \xstar \nr_F \\
	2) \quad &\nl \1 \ub^T \alpha \left(\Lambdab_{\barg} - \Lambdab_{\y}\right) \nr_F  \\
	& \leq \alpha \delta_{F, u}L_v \nl \1 \ub^T \nr_2 \nl \x(k) - \barx(k) \nr_F^{\ub} \\
	3) \quad &\nl \1 \alpha \ub^T \left( \Lambdab_{\y} - \Y(k) \right) \nr_F \\
	&\leq \alpha \nl \1 \ub^T \nr_2 \delta_{F, v} \sum_{h = 1}^H \nl \vb^h \bary(k) - \y^h(k) \nr_F^{\vb^h} \\
	4) \quad	&\nl \Y(k) \nr_F \leq  \sum_{h=1}^H \nl\y^h(k) - \vb^h\bary^h(k)\nr_F \\
	&+   \sqrt{H} \delta_{F, u} L_v\nl\x(k) - \1\barx(k) \nru \\
	&+   \sqrt{H} L_v \nl\1 \barx(k) - \1\xstar \nr_F \\
	\end{align*}
\end{lemma}
\begin{proof}
	Part 1): \\
	Expanding the square of the norm results in 
	\begin{align*}
	&\nl \1 \barx(k) - \1 \xstar -  \1 \ub^T \alpha\Lambdab_{\barg} \nr_F^2 \\
	&= \nl \1 \barx(k) - \1 \xstar \nr_F^2 \\
	&- 2 \left\langle \1 \barx(k) - \1 \xstar,  \1 \ub^T \alpha \diag\lbrace\vb^1 \barg^1(k), ..., \vb^H \barg^H(k)\rbrace\right\rangle_F\\
	& + \nl  \1 \ub^T \alpha \diag\lbrace\vb^1 \barg^1(k), ..., \vb^H \barg^H(k)\rbrace \nr_F^2.
	\end{align*}
	Here $\langle \cdot \rangle_F$ is the Frobenius inner product, see section \ref{subsec:matrixnorms} on matrix norms, which can be reformulated as follows using Assumption \ref{as:strongmonotonegame} and optimality condition \eqref{eq:socialwelfarecondition}:
	\begin{align*}
	&\left\langle \1 \barx(k) - \1 \xstar,  \alpha \1 \ub^T \diag\lbrace\vb^1 \barg^1(k), ..., \vb^H \barg^H(k)\rbrace\right\rangle_F \\
	&= \tr\left[n\alpha \diag\lbrace\vb^1 \barg^1(k), ..., \vb^H \barg^H(k)\rbrace^T \ub  \left( \barx(k) - \xstar\right) \right]\\
	&= n \alpha \sum_{h=1}^H \Big[  (\vb^h)^T \ub^h \barg^h(k) (\barx^h(k) - (\xstar)^h)^T \Big] \\
	&\geq  n \alpha \underline{\eta} \sum_{h=1}^H \Bigg[  \sum_{i = 1}^{n_h} (\nabla_{\barx^h}f_i^h(\barx^h(k), \barx^{-h}(k))- \nabla_{\x^h}f_i^h(\xstar))^T \\
	& \qquad(\barx^h(k) - (\xstar)^h)^T  \Bigg]\\
	& \geq n \alpha \underline{\eta} \mu ||\barx(k) - \xstar||_2^2 = \alpha \underline{\eta} \mu \nl\1 \barx(k) - \1 \xstar \nr_F^2,
	\end{align*}
	where it can be assumed that a  $\underline{\eta} >0$ exists such that $\eta^h = \frac{(\vb^h)^T \ub^h}{n_h} \geq \underline{\eta}, \ \forall h$. \\
	Furthermore, using Lemma \ref{lemma:submult}, Assumption \ref{as:lipschitz} and again the optimality condition \eqref{eq:socialwelfarecondition}, we receive
	\begin{align*}
	& \nl \alpha \1 \ub^T\diag\lbrace\vb^1 \barg^1(k), ..., \vb^H \barg^H(k)\rbrace \nr_F^2 \\
	& \leq  \alpha^2 \nl \1 \ub^T \nr_2^2 \Bigg(\sum_{h = 1}^H \sum_{i = 1}^{n_h} \frac{(v_i^h)^2}{n_h}  \Big|\Big| \sum_{i = 1}^{n_h} \Big(\nabla_{\barx^h}f_i^h(\barx(k)) \\
	&- \nabla_{\x^h}f_i^h(\xstar)\Big)^T \Big|\Big|_2^2 \Bigg) \\
	& \leq \alpha^2 L_v^2 \nl \1 \ub^T \nr_2^2 \nl \1 \barx(k) - \1 \xstar \nr_F^2
	\end{align*}
	with $L_v = L \max_{i,h}\lbrace  v_i^h\rbrace$. \\
	Summing up all results above and taking the square root concludes Part 1).\\
	Part 2): Applying  Lemma \ref{lemma:submult} and Assumption \ref{as:lipschitz}, we receive
	\begin{align*}
	&\nl \1 \ub^T \alpha \left(\Lambdab_{\barg} - \Lambdab_{\y}\right) \nr_F  \\
	& \leq \alpha \nl \1 \ub^T \nr_2 \Bigg(\sum_{h = 1}^H \sum_{i = 1}^{n_h} \frac{(\vb_i^h)^2}{n_h}  \Big|\Big| \sum_{i = 1}^{n_h}\Big(\nabla_{\barx^h}f_i^h(\barx(k)) \\
	&-  \nabla_{\x_i^h} f_i^h(\x_i(k)) \Big)\Big|\Big|_2^2 \Bigg)^{\frac{1}{2}} \\
	& \leq \alpha L_v \nl \1 \ub^T \nr_2 \Bigg(\sum_{h = 1}^H \sum_{i = 1}^{n_h}  ||\x_i(k) - \barx(k)||_2^2 \Bigg)^{\frac{1}{2}} \\
	& \leq \alpha \delta_{F, u}L_v \nl \1 \ub^T \nr_2 \nl \x(k) - \1 \barx(k) \nr_F^{\ub},
	\end{align*}
	where we used $L_v = L \max_{i,h}\lbrace  (v_i^h)^2\rbrace$ and the last inequality is true due to the equivalence relation of norms, see section \ref{subsec:matrixnorms}. \\
	Part 3): By applying  Lemma \ref{lemma:submult}, using equivalence of norms (see section \ref{subsec:matrixnorms}) and rewriting the Frobenius norm of the block matrix, we receive: 
	\begin{align*}
	&\nl \1 \alpha \ub^T \left( \Lambdab_{\y} - \Y(k) \right) \nr_F \\
	& \leq \alpha \nl \1 \ub^T \nr_2 \left( \sum_{h = 1}^H \nl \vb^h \bary^h(k) - \y^h(k) \nr_F^2 \right)^{\frac{1}{2}} \\
	& \leq \alpha \nl \1 \ub^T \nr_2 \delta_{F, v} \sum_{h = 1}^H \nl \vb^h \bary^h(k) - \y^h(k) \nr_F^{\vb^h}. 
	\end{align*}
	\newline
	Part 4): Due to the triangle inequality, it holds that
	\begin{align*}
	&\nl\Y(k)\nr_F = \sqrt{\sum_{h=1}^H \nl\y^h(k)\nr_F^2} \leq \sum_{h=1}^H \nl\y^h(k)\nr_F\\
	& \leq  \sum_{h=1}^H \left(\nl\y^h(k) - \vb^h\bary^h(k)\nr_F + \nl   \vb^h\bary^h(k)\nr_F \right).
	\end{align*}
	By using relation \eqref{eq:yinduction}, optimality condition \eqref{eq:socialwelfarecondition} and Assumption \ref{as:lipschitz}, we can provide the upper bound
	\begin{align*}
	& \nl   \vb^h\bary^h(k)\nr_F \\ 
	&\leq \Bigg(\sum_{i = 1}^{n_h} \frac{(v^h_i)^2}{n_h}   \sum_{i = 1}^{n_h} \Big|\Big| \nabla_{\x_i^h} f_i^h(\x_i^h(k), \x_i^{-h}(k))\\ 
	& \qquad - \nabla_{\x^h} f_i^h(\xstar)\Big|\Big|_2^2 \Bigg)^{\frac{1}{2}} \\
	&\leq \left( L_v^2 \sum_{i = 1}^{n_h} \frac{  1}{n_h} \left|\left| (\x_i - \xstar)_{i=1}^{n_h} \right|\right|_F^2 \right)^{\frac{1}{2}} \\
	&=  L_v \nl(\x_i - \barx(k) + \barx(k) -  \xstar)_{i=1}^{n_h}  \nr_F \\
	& \leq L_v\nl( \x_i(k) - \barx(k))_{i=1}^{n_h}  \nr_F \\
	&+   L_v \nl \1_{n_h}\barx(k) - \1_{n_h} \xstar \nr_F
	\end{align*}
	using $L_v =  L  \max_{i,h}\lbrace(v_i^h)^2\rbrace$.
	Plugging this into the result above and applying the norm equivalence relations, we receive
	\begin{align*}
	\nl\Y(k)\nr_F &\leq \sum_{h=1}^H \nl\y^h(k) - \vb^h\bary^h(k)\nr_F \\
	&+ \sum_{h=1}^H   L_v\nl( \x_i(k) - \barx(k))_{i=1}^{n_h} \nr_F \\
	&+  \sum_{h=1}^H  L_v \nl \1_{n_h}\barx(k) - \1_{n_h} \xstar \nr_F \\
	& \leq  \sum_{h=1}^H \nl\y^h(k) - \vb^h\bary^h(k)\nr_F \\
	&+   \sqrt{H} \delta_{F, u} L_v\nl\x(k) - \1\barx(k) \nru \\
	&+   \sqrt{H} L_v \nl\1 \barx(k) - \1\xstar \nr_F 
	\end{align*}
\end{proof}

\subsection{Proof of Proposition \ref{prop:matrixA}} \label{ap:proposition}
\begin{proof}
	For this proof, we consider each dimension of the vector separately and summarize the results at the end in matrix $\A$. \\
	\underline{Line 1:}
	We insert update equation of \eqref{alg:vector_x} into  $\barx(k+1) = \ub^T \x(k+1)$, add $ \1 \ub^T \alpha \Lambdab_{\y}  - \1 \ub^T \alpha  \Lambdab_{\y}$ and $\1 \ub^T \alpha \Lambdab_{\barg} - \1 \ub^T \alpha \Lambdab_{\barg}$ and apply the triangle inequality:
	\begin{align*}
	&\nl \1 \barx(k+1)- \1\xstar  \nr_F \\
	&\leq \nl \1 \barx(k)- \1 \xstar  -  \1 \ub^T \alpha \Lambdab_{\barg} \nr_F \\
	&+ \nl \1 \ub^T \alpha \left(\Lambdab_{\barg} - \Lambdab_{\y}\right) \nr_F + \nl \1  \ub^T\alpha \left( \Lambdab_{\y} - \Y(k) \right) \nr_F
	\end{align*}
	Applying the results in 1), 2) and 3) of Lemma \ref{lemma:inequalities}, respectively, we receive  the factors
	\begin{align*}
		\phi(\alpha) &= \sqrt{1  - 2\alpha \underline{\eta} \mu +  \alpha^2 L_v^2 \nl \1 \ub^T \nr_2^2}, \\
		a_{12} &= \delta_{F, u}L_v \nl \1 \ub^T \nr_2 ,	a_{13} = \alpha \nl \1 \ub^T \nr_2 \delta_{F, v}.
	\end{align*}
	\underline{Line 2:}
	Inserting the update equation of \eqref{alg:vector_x}, reordering and applying the triangle inequality leads to
	\begin{align*}
	&\nl \x(k+1) - \1 \barx(k+1) \nru \\
	& \leq \nl  \R\x(k) - \1 \barx(k)) \nru  + \alpha \delta_{u,F} \nl \left( \R  - \1 \ub^T \right) \Y(k) \nr_F \\
	\end{align*}
	Applying Lemma \ref{lemma:sigma} to the first element, Lemma $\ref{lemma:submult}$ to the second and
	inserting the result 4) of Lemma \ref{lemma:inequalities} for $\nl \Y(k) \nr_F$ yields the factors
	\begin{align*}
		a_{21} &= \delta_{u,F} \sigma_2 \sqrt{H}  L_v, a_{22} =    \delta_{u,F} \sigma_2  \sqrt{H} \delta_{F, u} L_v\\
		a_{23} &= \delta_{u,F} \delta_{F,v}\sigma_2,
	\end{align*}
	where  $\sigma_2 = \nl \R  - \1 \ub^T \nr_2 $.\\
	\underline{Line 3: }
	We expand the third line according to update equation \eqref{alg:vector_y} and then apply the triangle inequality
	\begin{align*}
	&\sum_{h=1}^H\nl \y^h(k+1) - \vb^h \bary^h(k+1)\nrv \\
	& \leq \sum_{h=1}^H  \nl  \C^h \y^h  -  \vb^h\bary^h(k) \nrv \\
	& +  \delta_{v, F}\sigma_I \sum_{h=1}^H  \nl \left(\G^h(k+1) - \G^h(k)\right) \nr_F
	\end{align*}
	using  $\sigma_I =\max_h \lbrace \nl I -  \vb^h \1^T \nr_F \rbrace$.
	The last expression can then be reformulated using Assumption \eqref{as:lipschitz}
	\begin{align*}
	& \sum_{h=1}^H  \nl \G^h(k+1) - \G^h(k) \nr_F \\
	& \leq \sum_{h=1}^H  \Bigg( \sum_{i = 1}^{n_h} L_i^h ||\x_i(k+1) - \x_i(k)||_2^2 \Bigg)^{\frac{1}{2}} \\
	& \leq L \sum_{h=1}^H  \nl (\x_i(k+1) - \x_i(k) )_{i=1}^{n_h}\nr_F \\
	& \leq L \sqrt{H} \nl \R \x(k) - \alpha\R\Y(k) - \x(k) \nr_F \\
	& \leq L \sqrt{H} \nl \R - \bm I\nr_2 \nl \x(k) - \1 \barx(k) \nr_F \\
	&+ \alpha L \sqrt{H} \nl \R \nr_2 \nl \Y(k) \nr_F. \\
	\end{align*}
	Inserting inequality 4) for $\nl \Y(k) \nr_F$ from Lemma \ref{lemma:inequalities} and plugging the result into above expression, we finally receive
	the factors
	\begin{align*}
	 &a_{31}  = \delta_{v, F}\sigma_I L H   L_v , a_{32}  =  \delta_{v, F}\sigma_I \delta_{F, u} L \sqrt{H}\nl \R - \bm I\nr_2, \\
	&a_{32'}  =  \delta_{v, F}\sigma_I  L H  \delta_{F, u} L_v, a_{33} = \delta_{v, F} \sigma_I L \sqrt{H} \delta_{F, v}.
	\end{align*}
	and $\nl \R \nr_2$ = 1, because of Lemma \ref{lemma:RCeigen}. 
\end{proof}

\end{document}